\documentclass[12pt,a4paper]{article}
\usepackage[dvipdfmx,hiresbb]{graphicx}
\usepackage[dvipdfmx]{color}

\usepackage{color}
\usepackage{bm,enumerate,amsmath,amssymb,amsthm}
\usepackage{epsfig}
\usepackage{cite}
\usepackage{algorithm}
\usepackage{algorithmic}
\usepackage{txfonts}
\usepackage{subfigmat}
\usepackage{figsize}
\usepackage{here}
\usepackage{listings}
\usepackage{braket}
\usepackage{comment}
\usepackage{lscape}
\usepackage{bigints}

\usepackage{arxiv}
\usepackage[utf8]{inputenc} 

\theoremstyle{definition}
\newtheorem{theorem}{Theorem}
\newtheorem*{theorem*}{Theorem}
\newtheorem{definition}{Definition}
\newtheorem*{definition*}{Definition}

\newtheorem{lemma}{Lemma}

\newtheorem{corollary}{Corollary}

\title{Generalization Analysis of Quantum Neural Networks Using Dynamical Lie Algebras}

\author{
  Hiroshi Ohno \\
  Toyota Central R \& D Labs., Inc.\\
  Aichi, Japan \\
  \texttt{oono-h@mosk.tytlabs.co.jp} \\
}

\date{\empty}

\begin{document}
\maketitle

\begin{abstract}
  The paper presents a generalization bound for quantum neural networks based on a dynamical Lie algebra.
  Using covering numbers derived from a dynamical Lie algebra, the Rademacher complexity is derived to calculate the generalization bound.
  The obtained result indicates that the generalization bound is scaled by $ \mathcal{O}( \sqrt{ \dim( \mathfrak{g}) } ) $, where $ \mathfrak{g} $ denotes a dynamical Lie algebra of generators.
  Additionally, the upper bound of the number of the trainable parameters in a quantum neural network is presented.
  Numerical simulations are conducted to confirm the validity of the obtained results.
\end{abstract}

\section{Introduction}
In quantum machine learning, generalization is a crucial capability of quantum neural networks (or parameterized quantum circuits), and its theoretical analysis, based on (classical) statistical machine learning theory, has been actively explored.
In this study, we present a generalization bound for quantum neural networks using covering numbers \cite{mohri2018} derived from dynamical Lie algebras (DLAs) \cite{ragone2024} and demonstrate numerical simulation results related to it.
Our main contribution is that the generalization bound is scaled as $ \mathcal{O}( \sqrt{ \dim( \mathfrak{g}) } ) $, where $ \mathfrak{g} $ denotes a DLA of generators.

There are several theoretical results regarding generalization bounds.
For models with moderate complexity, Banchi et al. \cite{banchi2021} showed that the generalization error scales as $ \mathcal{O} (\sqrt{ \frac{\mathcal{B}}{M} } )$, where $ \mathcal{B} $ represents a quantity based on the R\'{e}nyi mutual information between the quantum space embedded input data and the original input data space, and $ M $ denotes the training data size.
Caro et al. \cite{caro2022} demonstrated that the generalization error (or gap) scales as $ \mathcal{O} (\sqrt{ \frac{N_{t} \log N_{t}}{M} } ) $, where $ N_{t} $ represents the number of trainable parameters.
They used the metric entropy (i.e., the logarithm of covering numbers) to quantify model complexity.
From the perspective of PAC learnability, Cai et al. \cite{cai2022} established that the sample complexity is bounded by $ \mathcal{O}(n^{c+1}) $, where $ n^{c} $ denotes the number of unitary gates, $ n $ represents the number of qubits, and $ c \geq 0 $ is a parameter that can be related to circuit depth.
Yuxuan et al. \cite{yuxuan2022} showed that, using covering numbers, the generalization gap is bounded by $ \tilde{\mathcal{O}} (\frac{\sqrt{N_{t}}}{\sqrt{M}}) $.
To the best of our knowledge, this is the first work to introduce DLA into the assessment of the generalization capability of quantum neural networks.
Additionally, an upper bound on the number of trainable parameters is presented, and related numerical simulations are conducted to verify our theoretical results.

Recently, the use of DLA has been expanding for analyzing the learning properties of quantum neural networks.
The analysis of barren plateaus (BPs) \cite{ragone2024} and the overparameterization phenomenon \cite{larocca2023} in quantum neural networks has also been explored using DLA.
Ragone et al. \cite{ragone2024} demonstrated the relationship between the dimension of the DLA and the variance of the loss function with respect to the parameters.
Larocca et al. \cite{larocca2023} provided a sufficient condition for the number of parameters required for overparameterization using the dimension of the DLA.
Hence, DLA is a crucial framework (or tool) for understanding the learning properties of quantum neural networks.

\section{Related work}
In this section, we briefly review some previous studies that have significantly influenced our work and point out issues in their approaches.
These reviews aim to develop a generalization bound using covering numbers.

Yuxuan et al. \cite{yuxuan2022} proposed a method for measuring the expressivity of variational quantum algorithms, such as quantum neural networks, using covering numbers.
In their study, they also derived a generalization upper bound, which corresponds to the upper bound of the generalization gap, as referred to in this study.
The implications of this theoretical result are as follows:
\begin{itemize}
\item The generalization bound has an exponential dependence with the largest number of qudits operated on a single gate and a sublinear dependence
  with the number of trainable gates (parameters).
\item Increasing the training data size leads to an improvement of the generalization performance.
\item The sublinear dependence with the number of trainable gates limits to assess the generalization performance for overparameterization.
\end{itemize}
In the proof of Lemma 2 in \cite{yuxuan2022}, the authors used Lemma 1, $ \left( \frac{3}{4\tilde{\varepsilon}} \right)^{n^{2}} \leq \mathcal{N}(U(n), \| \cdot \|, \tilde{\varepsilon}) \leq \left( \frac{7}{\tilde{\varepsilon}} \right)^{n^{2}} $, in \cite{barthel2018},  which provides bounds for the covering numbers $ \mathcal{N}( \cdot ) $ of the unitary group $ U(n) $, where $ \tilde{\varepsilon} $ denotes the radius of the ball.
However, the equality on the right-hand side of Lemma 1 does not hold due to a precondition $ \tilde{\varepsilon} \leq 0.1 $.
Additionally, in Eq. (A7) of \cite{yuxuan2022}, when applying the Cauchy-Schwartz inequality, the Frobenius norm should be used instead of the operator norm.
As a result, the corresponding equation requires modification.
Furthermore, the authors did not establish a relationship between $ \tilde{\varepsilon} $ and the number of trainable parameters, which raises concerns about the validity of Lemma 2.

\section{Preliminaries}
We present the definition of DLA and definitions relevant to DLA \cite{larocca2022}.

\begin{definition} (Set of generators)
  The set of generators $ \mathcal{G} $ is defined by
  \begin{equation*}
    \mathcal{G} \coloneqq \{ H_{k} \},
  \end{equation*}
  where $ H_{k} $ denotes the traceless Hermitian matrices that generate the unitary matrices and $ | \mathcal{G} | < \infty $.
\end{definition}
Since $ H_{k} $ are the traceless Hermitian matrices, the generated unitary matrices are special unitary matrices.

\begin{definition} (DLA)
  Given $ \mathcal{G} $, DLA $ \mathfrak{g} $ is defined by
  \begin{equation*}
    \mathfrak{g} \coloneqq {\rm span}_{\mathbb{R}} \braket{ iH_{0}, \ldots, iH_{K} }_{Lie},
  \end{equation*}
  where $ K = | \mathcal{G} | $, $ {\rm span}_{\mathbb{R}} (\mathcal{A}) $ is the set of all linear combinations of elements of $ \mathcal{A} $ with real coefficients, and $ \braket{\mathcal{S}}_{Lie} $ denotes the Lie closure, that is , the set obtained by repeatedly taking the nested commutators between the elements in $ \mathcal{S} $\footnote{
    If a quantum system has a finite (infinite) size, then the resulting set also has a finite (infinite) size.
  }.
\end{definition}
$ \mathfrak{g} $ is a subalgebra of $ \mathfrak{su}(d) $ which is the special unitary algebra of size $ d $ whose element is $ d \times d $ traceless skew-Hermitian matrices\footnote{
  When a DLA decomposes into a direct sum of multiple simple algebras and an Abelian center (as shown in Eq. (6) of \cite{ragone2024}, Proposition A.1 in \cite{wiersema2024}), the unitary group it generates becomes structured and fragmented.
  The number of simple components determines how many independent block operations constitute each overall unitary.
  As the number of simple components increases, the number of block operations also grows, leading to a decomposition of the Hilbert space into smaller, distinct invariant subspaces.
  Moreover, the Haar randomness of the unitary matrix generated by the DLA over the entire Hilbert space may diminish.
  Consequently, this reduction in randomness limits the expressivity of the unitary.
  Therefore, more entanglement layers may be required.
}.
More formally, a bilinear operation $ [a, b] \in \mathfrak{g} $ is satisfied by the following skew-symmetry and the Jacobi identity: $ [a, b] = - [b, a] $ and $ [a, [b, c]] = [[a, b], c] + [b, [a, c]] $.
We define $ {\rm ad}_{a}(b) \coloneqq [a, b] $ for $ a, b \in \mathfrak{g} $.
For a generator set $ \mathcal{G} = \{g_{1}, \ldots, g_{K}\} $, the nested commutator is defined as follows:
\begin{equation}
  {\rm ad}_{i_{1}} \cdots {\rm ad}_{i_{r}} (g_{j}) \coloneqq [g_{i_{1}}, [g_{i_{2}}, [\cdots [g_{i_{r}}, g_{j}] \cdots]]],
\end{equation}
where the nested commutator is $ g_{j} $ for $ r = 0 $.
Then, DLA $ \mathfrak{g} $ is defined as follows \cite{wiersema2024}:
\begin{equation}
  \mathfrak{g} = \braket{\mathcal{G}}_{{\rm Lie}} \coloneqq i \, {\rm span}_{\mathbb{R}} \left\{ {\rm ad}_{g_{i_{1}}} \cdots {\rm ad}_{g_{i_{r}}} (g_{j}) \; \middle| \; g_{i_{1}}, \ldots, g_{i_{r}}, g_{j} \in \mathcal{G} \right\}.
\end{equation}

Additionally, we define the dynamical Lie group according to DLA.
\begin{definition} (Dynamical Lie group)
  The set unitaries $ G $ is determined by its DLA as follows:
  \begin{equation*}
    G = \exp( \mathfrak{g} ) \coloneqq \left\{ \exp( g ) \; \middle| \; g \in \mathfrak{g} \right\},
  \end{equation*}
  where $ \exp $ denotes the matrix exponential, $ \exp ( X ) \coloneqq \sum_{k = 0}^{\infty} \frac{X^{k}}{k!} $.\footnote{
    If $ X $ is a $ d \times d $ matrix, then the series always converges.
  }
\end{definition}
We note that $ \exp(\cdot) $ is not a Lipschitz function.
Hence, when $ \exp(\cdot) $ is used as a Lipschitz function, constraints such as upper bounds on its norm must be appropriately imposed on its argument.

\section{Methods}
In this section, at the beginning, we summarize useful lemmas to develop the main result (Theorem \ref{thm1}).

\subsection{Useful lemmas}
\begin{lemma}(Lemma 3 in \cite{barthel2018})\label{lem1}
  Let $ (\mathcal{H}_{1}, d_{1}) $ and $ (\mathcal{H}_{2}, d_{2}) $ be metric spaces, where $ d_{1} $ and $ d_{2} $ are metrics on $ \mathcal{H}_{1} $ and $ \mathcal{H}_{2} $, respectively and $ f: \mathcal{H}_{1} \longrightarrow \mathcal{H}_{2} $ be a bi-Lipschitz function such that
  \begin{equation}
    d_{2} (f(x), f(y)) \leq K d_{1} (x, y), \forall x, y \in \mathcal{H}_{1}
  \end{equation}
  and
  \begin{equation}
    k d_{1} (x, y) \leq d_{2} (f(x), f(y)), \forall x, y \in \mathcal{H}_{1} \; {\rm with} \; d_{1}(x, y) \leq r,
  \end{equation}
  where $ K $ and $ k $ are constant.
  Then, their covering numbers obey
  \begin{equation}
    \mathcal{N}(\mathcal{H}_{1}, d_{1}, 2\varepsilon/k) \leq \mathcal{N}(\mathcal{H}_{2}, d_{2}, \varepsilon) \leq \mathcal{N}(\mathcal{H}_{1}, d_{1}, \varepsilon/K),
  \end{equation}
where the left inequality requires $ 0 < \varepsilon \leq kr/2 $.
\end{lemma}

\begin{lemma}(Appendix D in \cite{barthel2018})\label{lem2}
  For the matrix exponential $ \exp(X) $ and $ \exp(Y) $, where $ X, Y \in $ a Lie algebra, when $ \| X \|_{{\rm op}} \leq p $ and $ \| Y \|_{{\rm op}} \leq p $, then the following inequality holds:
  \begin{equation}
    (2 - \exp(p)) \, \| X - Y \|_{{\rm op}} \leq \| \exp(X) - \exp(Y) \|_{{\rm op}} \leq \| X - Y \|_{{\rm op}},
  \end{equation}
  where $ \| \cdot \|_{{\rm op}} $ denotes the operator norm.
\end{lemma}

\begin{lemma}\label{lem3}
  For Hermitian matrices $ X $ and $ Y $ and a density matrix $ \rho $, the following inequality holds:
  \begin{equation}
    \left | \, {\rm Tr}[ X \rho ] - {\rm Tr}[ Y \rho ] \, \right | \leq \sqrt{N} \, \| X - Y \|_{{\rm op}},
  \end{equation}
  where $ N $ denotes the number of eigenvalues of $ X - Y $.
\end{lemma}
\begin{proof}
  \begin{equation}
    \begin{split}
      &\left | \, {\rm Tr}[ X \rho ] - {\rm Tr}[ Y \rho ] \, \right | = \left | \, {\rm Tr}[ (X - Y)\rho ] \, \right |\\
      &= \left | \, {\rm Tr}[ Z \rho ] \, \right |\\
      &= \left | \, \langle Z, \rho \rangle \, \right |\\
      &\leq \| Z \|_{{\rm F}} \, \| \rho \|_{{\rm F}}\\
      &= \| Z \|_{{\rm F}},
    \end{split}
  \end{equation}
  where in the second equality, $ Z = X - Y $, the third equality uses the inner product $ \langle \cdot, \cdot \rangle $ and $ \rho^{\dagger} = \rho $, the first inequality uses the Cauchy-Schwartz inequality and $ \| \cdot \|_{{\rm F}} $ denotes the Frobenius norm, and the last equality is the fact that $ \| \rho \|_{{\rm F}} = \sqrt{ \langle \rho, \rho \rangle } = \sqrt{ {\rm Tr} [ \rho^{2} ] } = 1 $.
  From the definition of the Frobenius norm, $ \| Z \|_{{\rm F}} = \sqrt{ \sum_{j = 1}^{N} \sigma^{2}_{j} } $, where $ \sigma_{j} $ is the eigenvalues of $ Z $, and $ \sigma_{max} = \max_{j \in [N]} \sigma_{j} $.
  Then, we obtain $ \| Z \|_{{\rm F}} = \sqrt{ \sum_{j = 1}^{N} \sigma^{2}_{j} } \leq \sqrt{ \sum_{j = 1}^{N} \sigma^{2}_{max} } = \sqrt{N} \sigma_{max} $.
  From the definition of the operator norm, $ \| Z \|_{{\rm F}} \leq \sqrt{N} \, \| Z \|_{{\rm op}} $, thus, we obtain
  \begin{equation}
    \left | \, {\rm Tr}[ X \rho ] - {\rm Tr}[ Y \rho ] \, \right | \leq \sqrt{N} \, \| X - Y \|_{{\rm op}}.
  \end{equation}
\end{proof}

\begin{lemma}\label{lem4}
  For unitary matrices $ U $ and $ V $ and an observable (Hermitian matrix) $ O $, the following inequality holds:
  \begin{equation}
    \left \| U^{\dagger} O U - V^{\dagger} O V \right \|_{{\rm op}} \leq 2 \| U - V \|_{{\rm op}} \, \| O \|_{{\rm op}}.
  \end{equation}
\end{lemma}
\begin{proof}
  \begin{equation}
    \begin{split}
      &\left \| U^{\dagger} O U - V^{\dagger} O V \right \|_{{\rm op}}\\
      &= \left \| U^{\dagger} O U - V^{\dagger} O U + V^{\dagger} O U - V^{\dagger} O V \right \|_{{\rm op}}\\
      &\leq \left \| U^{\dagger} O U - V^{\dagger} O U \right \|_{{\rm op}} + \left \| V^{\dagger} O U - V^{\dagger} O V \right \|_{{\rm op}}\\
      &= \left \| (U^{\dagger} - V^{\dagger}) O U \right \|_{{\rm op}} + \left \| V^{\dagger} O (U - V) \right \|_{{\rm op}}\\
      &\leq \left \| U^{\dagger} - V^{\dagger} \right \|_{{\rm op}} \| O \|_{{\rm op}} \| U \|_{{\rm op}} + \| V^{\dagger} \|_{{\rm op}} \| O \|_{{\rm op}} \| U - V \|_{{\rm op}}\\
      &= \| U - V \|_{{\rm op}} (\| U \|_{{\rm op}} + \| V^{\dagger} \|_{{\rm op}} ) \| O \|_{{\rm op}}\\
      &= 2 \| U - V \|_{{\rm op}} \| O \|_{{\rm op}},
    \end{split}
  \end{equation}
  where the first inequality uses the triangle inequality, the second inequality uses the sub-multiplicative property of the operator norm, and the last equality uses the fact that the operator norm of the unitary matrix is one.
\end{proof}

\subsection{Covering number of quantum circuits}
For a DLA $ \mathfrak{g} $, we consider a ball $ \mathcal{B}_{\pi}( \mathfrak{g} ) \subset \mathbb{R}^{\dim( \mathfrak{g} )} $, where the radius of the ball is $ \pi $ and for $ g \in \mathfrak{g} $, $ \| g \|_{{\rm op}} \leq \pi $.
The Lie group $ U $ is generated by $ \exp( \mathfrak{g} ) $.
Then, for the covering number of $ \mathcal{B}_{\pi}( \mathfrak{g} ) $, the following inequality holds:
\begin{equation}\label{eq1}
  \mathcal{N}( \mathcal{B}_{\pi}(\mathfrak{g}), \| \cdot \|_{{\rm op}}, \varepsilon) \leq \left( 1 + \frac{2 \pi}{\varepsilon} \right)^{\dim( \mathfrak{g} )}.
\end{equation}

Using Lemmas \ref{lem1} and \ref{lem2} with $ \mathcal{H}_{1} = \mathcal{B}_{\pi}(\mathfrak{g}) $ and $ \mathcal{H}_{2} = U $, then $ K = 1 $, $ k = 2 - \exp(p) $, and the following inequality holds:
\begin{equation}\label{eq2}
  \mathcal{N}( U, \| \cdot \|_{{\rm op}}, \varepsilon) \leq \mathcal{N}( \mathcal{B}_{\pi}(\mathfrak{g}), \| \cdot \|_{{\rm op}}, \varepsilon) \leq \left( 1 + \frac{2 \pi}{\varepsilon} \right)^{\dim( \mathfrak{g} )},
\end{equation}
where the inequality on the left-hand side holds when $ \varepsilon \leq \frac{k r}{2} = \frac{(2 - \exp(p)) 2p }{2} = (2 - \exp(p))p $ because of $ r = 2p $.\footnote{
  From Lemma \ref{lem2}, when $ \| X \|_{{\rm op}} \leq p $ and $ \| Y \|_{{\rm op}} \leq p $, $ \| X - Y \|_{{\rm op}} \leq 2p $.
}
Thus, we obtain the following inequality:
\begin{equation}\label{eq3}
  \mathcal{N}( U, \| \cdot \|_{{\rm op}}, \varepsilon) \leq \left( 1 + \frac{2 \pi}{\varepsilon} \right)^{\dim( \mathfrak{g} )}.
\end{equation}

Next, we consider $ \mathcal{H} = \{ {\rm Tr}[ X \rho ] \} $ and $ \mathcal{H}_{circ} = \{ X \} $, where $ X $ denotes Hermitian matrices and $ \rho $ denotes density matrices.
Using Lemma \ref{lem3}, we obtain the following inequality for their covering numbers:
\begin{equation}\label{eq4}
  \mathcal{N}( \mathcal{H}, | \cdot |, \varepsilon) \leq \mathcal{N}( \mathcal{H}_{circ}, \| \cdot \|_{{\rm op}}, \varepsilon/\sqrt{N}),
\end{equation}
where $ N $ is the number of eigenvalues of $ X $.

As a quantum circuit, we define a parameterized unitary matrix $ U (\theta) $ as follows:
\begin{equation}\label{eq5}
  U(\theta) \coloneqq \prod_{l = 1}^{L} u(\theta_{l}) = \prod_{l = 1}^{L} \prod_{k = 1}^{K} \exp( g_{k} \theta_{l,k} ),
\end{equation}
where $ L $ denotes the circuit length, $ g_{k} \in \mathfrak{g} $, and $ \theta_{l,k} $ denotes trainable parameters.
Here, we divide the $ u $ into $ N_{t} $ trainable unitary matrices and $ N_{f} $ fixed unitary matrices.
In this study, the fixed unitary matrices are corresponding to entanglement circuits which provide the entanglement into the circuit.
Then, we rewrite $ U(\theta) $ as follows:
\begin{equation}\label{eq6}
  U(\theta) = \prod_{l = 1}^{L} {\rm I}[l \in {\rm L}_{{\rm fixed}} ] u_{l} + {\rm I}[l \in {\rm L}_{{\rm trainable}} ] u_{l}(\theta_{l}),
\end{equation}
where $ {\rm I} [\cdot] $ denotes the indicator function, $ {\rm L}_{{\rm trainable}} $ denotes a set of trainable gate number whose size is $ N_{t} $, and $ {\rm L}_{{\rm fixed}} $ denotes a set of fixed gate number whose size is $ N_{f} $.
Here, for the quantum circuits, we consider $ \mathcal{H}_{circ} = \{ U(\theta)^{\dagger} O U(\theta) \mid \theta \in \Theta \} $ and a covering set $ \mathcal{S}_{\varepsilon} $ for $ \mathcal{N}( U, \| \cdot \|_{{\rm op}}, \varepsilon) $, where $ \Theta $ is the parameter space.
In addition, we define the set
\begin{equation}\label{eq7}
  \tilde{ \mathcal{S} }_{\varepsilon} \coloneqq \left\{ \prod_{l = 1}^{L} {\rm I}[l \in {\rm L}_{{\rm fixed}} ] u_{l} + {\rm I}[l \in {\rm L}_{{\rm trainable}} ] u_{l}(\theta_{l}) \; \middle| \; u_{l}(\theta) \in \mathcal{S}_{\varepsilon} \right\}.
\end{equation}
Using Lemma \ref{lem4} for $ U(\theta) $ and $ U_{\varepsilon}(\theta) \in \tilde{ \mathcal{S} }_{\varepsilon} $, we obtain the following inequality
\begin{equation}\label{eq8}
  \left \| U(\theta)^{\dagger} O U(\theta) - U_{\varepsilon}(\theta)^{\dagger} O U_{\varepsilon}(\theta) \right \|_{{\rm op}} \leq 2 \| U(\theta) - U_{\varepsilon}(\theta) \|_{{\rm op}} \, \| O \|_{{\rm op}}.
\end{equation}
Here, considering $ \| U(\theta) - U_{\varepsilon}(\theta) \|_{{\rm op}} \leq (N_{t} - 1) \varepsilon $,\footnote{
  For example, for $ U_{\varepsilon} $, $ U^{'}_{\varepsilon} $, $ U $, and $ U^{'} $, since $ \| U_{\varepsilon} - U^{'}_{\varepsilon} \|_{{\rm op}} = 2 \varepsilon $, $ \| U - U_{\varepsilon} \|_{{\rm op}} \leq \varepsilon $, and $ \| U^{'} - U^{'}_{\varepsilon} \|_{{\rm op}} \leq \varepsilon $,
  $ \| U - U^{'}_{\varepsilon} \|_{{\rm op}} \leq 3 \varepsilon $.
  Thus, $ \| U - U^{'}_{\varepsilon} \|_{{\rm op}} \leq (N_{t} - 1) \varepsilon $, since $ N_{t} = 4 $.
} the following inequality is obtained:
\begin{equation}\label{eq9}
  \left \| U(\theta)^{\dagger} O U(\theta) - U_{\varepsilon}(\theta)^{\dagger} O U_{\varepsilon}(\theta) \right \|_{{\rm op}} \leq 2 (N_{t} - 1) \varepsilon \| O \|_{{\rm op}}.
\end{equation}
From Eq. (\ref{eq3}), $ | \mathcal{S}_{\varepsilon} | \leq \left( 1 + \frac{2 \pi}{\varepsilon} \right)^{\dim( \mathfrak{g} )} $.
There are $ | \mathcal{S}_{\varepsilon} |^{N_{t}} $ combinations for the gates in $ \tilde{\mathcal{S}}_{\varepsilon} $.
Thus, $ | \tilde{\mathcal{S}}_{\varepsilon} | \leq \left( 1 + \frac{2 \pi}{\varepsilon} \right)^{N_{t}\dim( \mathfrak{g} )} $.
Together with Eq. (\ref{eq9}), we obtain a covering number of $ \mathcal{H}_{circ} $ as follows:
\begin{equation}\label{eq10}
  \mathcal{N}(\mathcal{H}_{circ}, \| \cdot \|_{{\rm op}}, 2 (N_{t} - 1) \, \| O \|_{{\rm op}} \, \varepsilon) \leq \left( 1 + \frac{2 \pi}{\varepsilon} \right)^{N_{t}\dim( \mathfrak{g} )}.
\end{equation}
Thus, the following covering number is obtained:
\begin{equation}\label{eq11}
  \mathcal{N}(\mathcal{H}_{circ}, \| \cdot \|_{{\rm op}}, \varepsilon) \leq \left( 1 + \frac{4 \pi (N_{t} - 1) \, \| O \|_{{\rm op}} }{\varepsilon} \right)^{N_{t}\dim( \mathfrak{g} )}.
\end{equation}
Finally, using Eq. (\ref{eq4}), we obtain the covering number of $ \mathcal{H} = \left \{ {\rm Tr}[ U(\theta)^{\dagger} O U(\theta) \rho ] \, \mid \, \theta \in \Theta \right \} $ as follows:
\begin{equation}\label{eq12}
  \mathcal{N}(\mathcal{H}, | \cdot |, \varepsilon) \leq \left( 1 + \frac{4 \pi (N_{t} - 1) \sqrt{N} \, \| O \|_{{\rm op}} }{\varepsilon} \right)^{N_{t}\dim( \mathfrak{g} )}.
\end{equation}

\subsection{Main result: Generalization bound by the Lie-algebra perspective}
The generalization bound for the quantum circuits is given by the Rademacher complexity $ \mathcal{R}(\mathcal{H}) $ with probability $ 1 - \delta $ as follows \cite{mohri2018}:
\begin{equation}\label{eq13}
  \mathcal{L}( h(\theta_{*} ) ) - \hat{ \mathcal{L} }( h(\theta_{*} ) ) \leq 2 \mathcal{R}(\mathcal{H}) + 3 C \sqrt{ \frac{\ln \frac{2}{\delta}}{2 M} },
\end{equation}
where $ h(\cdot) \in \mathcal{H} $, $ \mathcal{L}( \cdot ) $ denotes the expected risk $ \mathbb{E}_{(x,y) \sim P_{D}} [ \mathtt{l}(h(\theta_{*}), y) ] $, where $ \mathtt{l} $ is a risk (cost) function and $ P_{D} $ is an unknown probability distribution, $ \hat{ \mathcal{L} }( \cdot ) $ denotes the empirical risk $ \frac{1}{M} \sum_{i=1}^{M} \mathtt{l}(h(\theta_{*}), y_{i}) $, $ \theta_{*} $ is the parameter optimized by a training algorithm $ \mathcal{A} $ with training data, $ C $ is a positive constant, and $ M $ is the number of training data.

\begin{theorem}\label{thm1}
  The following generalization gap holds.
  \begin{equation}
    \mathcal{L}( h(\theta_{*} ) ) - \hat{ \mathcal{L} }( h(\theta_{*} ) ) = \mathcal{O} \left( \sqrt{\dim( \mathfrak{g} )} \right).
  \end{equation}
\end{theorem}
\begin{proof}
  Given a covering number $ \mathcal{H} $, $ \mathcal{R} $ is given by the Dudley entropy integral as follows:
  \begin{equation}\label{eq14}
    \mathcal{R}(\mathcal{H}) \leq \inf_{\alpha > 0} \left ( 4\alpha + \frac{12}{\sqrt{M}} \int_{\alpha}^{1} \sqrt{ \ln \mathcal{N}(\mathcal{H}, | \cdot |, \varepsilon) } \; d\varepsilon \right ).
  \end{equation}
  Using Eq. (\ref{eq12}), the integral part in Eq. (\ref{eq14}) is calculated as follows:
  \begin{equation}
    \begin{split}\label{eq15}
      &\int_{\alpha}^{1} \sqrt{ \ln \mathcal{N}(\mathcal{H}, | \cdot |, \varepsilon) } \; d\varepsilon\\
      &\leq \bigintsss_{\alpha}^{1} \sqrt{ \ln \left( 1 + \frac{4 \pi (N_{t} - 1) \sqrt{N} \, \| O \|_{{\rm op}} }{\varepsilon} \right)^{N_{t}\dim( \mathfrak{g} )} } \; d\varepsilon\\
      &= \sqrt{N_{t}\dim( \mathfrak{g} )} \bigintsss_{\alpha}^{1} \sqrt{ \ln \left( 1 + \frac{4 \pi (N_{t} - 1) \sqrt{N} \, \| O \|_{{\rm op}} }{\varepsilon} \right) } \; d\varepsilon\\
      &\leq \sqrt{N_{t}\dim( \mathfrak{g} )} \bigintsss_{\alpha}^{1} \ln \left( 1 + \frac{4 \pi (N_{t} - 1) \sqrt{N} \, \| O \|_{{\rm op}} }{\varepsilon} \right) \; d\varepsilon\\
      &= \sqrt{N_{t}\dim( \mathfrak{g} )} \left( \alpha \ln \alpha + (1 + D) \ln (1 + D) - (\alpha + D) \ln (\alpha + D) \right),
    \end{split}
  \end{equation}
  where $ D \coloneqq 4 \pi (N_{t} - 1) \sqrt{N} \, \| O \|_{{\rm op}} $.
  Using Eq. (\ref{eq13}) and $ \alpha = \frac{1}{\sqrt{M}} $, we obtain the following bound:
  \begin{equation}\label{eq16}
    \begin{split}
      &\mathcal{L}( h(\theta_{*} ) ) - \hat{ \mathcal{L} }( h(\theta_{*} ) )\\
      &\leq \frac{8}{\sqrt{M}} \left( 1 + 3 \sqrt{N_{t}\dim( \mathfrak{g} )} \left( \frac{1}{\sqrt{M}} \ln \frac{1}{\sqrt{M}} + (1 + D) \ln (1 + D) - \left(\frac{1}{\sqrt{M}} + D \right) \ln \left(\frac{1}{\sqrt{M}} + D \right) \right) \right)\\
      &+ 3 C \sqrt{ \frac{\ln \frac{2}{\delta}}{2 M} }.
    \end{split}
  \end{equation}
  Thus, we obtain the following result:
  \begin{equation}\label{eq17}
    \mathcal{L}( h(\theta_{*} ) ) - \hat{ \mathcal{L} }( h(\theta_{*} ) ) = \mathcal{O} \left( \sqrt{\dim( \mathfrak{g} )} \right).
  \end{equation}
\end{proof}
We note that the generalization gap depends on training algorithms we use.

Regarding the number of trainable parameters, we obtain the following corollary.
\begin{corollary}\label{coro1}
  When $ \exp{ (p) } < 2 $, then the number of trainable parameters $ N_{t} $ satisfies the following inequality:
  \begin{equation}\label{eq18}
    N_{t} \leq \frac{2}{(2 - \exp{(p)}) p} + 1.
  \end{equation}
\end{corollary}
\begin{proof}
  From Lemmas \ref{lem1} and \ref{lem2}, $ 0 < \varepsilon \leq (2 - \exp{(p)}) p $.
  Thus, $ 0 < 2 - \exp{(p)} $ and $ \| U(\theta) - U_{\varepsilon}(\theta) \|_{{\rm op}} \leq (N_{t} - 1) \varepsilon \leq (N_{t} - 1) (2 - \exp{(p)}) p $.
  Since $ \| U(\theta) - U_{\varepsilon}(\theta) \|_{{\rm op}} \leq 2 $, the following inequality must be satisfied.
  \begin{equation}\label{eq19}
    (N_{t} - 1) (2 - \exp{(p)}) p \leq 2.
  \end{equation}
  To complete the proof, since $ p > 0 $ and $ \exp{(p)} < 2 $, we rearrange the above equation.
\end{proof}
According to Corollary \ref{coro1}, for example, for $ p = 0.1 $, $ N_{t} < 23.4 $ and $ \varepsilon \leq 0.045 $.
Given $ p $ or $ \varepsilon $ value, $ N_{t} $ is upper-bounded as presented in Corollary \ref{coro1}.\footnote{
  For the case of $ \varepsilon $,
  \begin{equation}\label{eq20}
    N_{t} \leq \frac{2}{\varepsilon} + 1.
  \end{equation}
}

Figure \ref{fig-1} shows the relationship between $ N_{t} $ and $ p $ under $ p \in [0.1, 0.69] $.
\begin{figure}[htbp]
  \centering
  \begin{tabular}{c}
    \includegraphics[width=8cm]{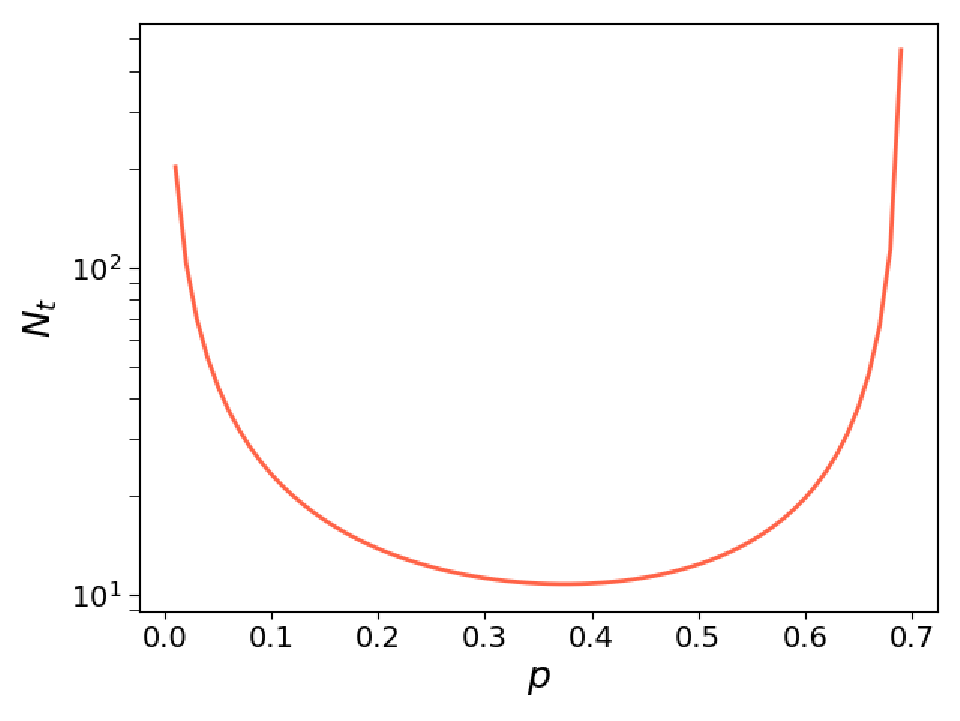}
  \end{tabular}
  \caption{Relationship between $ N_{t} $ and $ p $. $ p \in [0.1, 0.69] $ ($ \exp{(p)} < 2 $).}\label{fig-1}
\end{figure}
At $ p = 0.374 \cdots $, the minimum value of $ N_{t} $ is $ 10.785 \cdots $ regardless of the number of data $ M $.
However, the condition $ \exp{(p)} < 2 $ may be too severe in practice.

In the training algorithm $ \mathcal{A} $, the $ p $ value corresponds to the largest singular value of $ \theta \mathfrak{g} $ ($ \sigma_{max}(\theta \mathfrak{g}) $), where $ \theta $ is the trainable parameter.
Here, we determine $ \theta_{max} $ such that the largest singular value of $ \theta \mathfrak{g} $ ($ \theta \in \mathbb{R} $) is equal to $ \ln(2) $.

\section{Experiments}
Referring to \cite{yuxuan2022}, we constructed synthetic datasets to verify Theorem \ref{thm1}.\\
\noindent
{\bf Data construction.}
The synthetic dataset $ S = \{ x_{i}, y_{i} \} $ was constructed as follows:
\begin{equation}
  x_{i} \in [0, 2\pi]^{n},
\end{equation}
\begin{equation}
  y_{i} = \braket{ 0^{n} |\, U_{E}(x_{i})^{\dagger} V^{\dagger} O V U_{E}(x_{i}) \,| 0^{n} },
\end{equation}
where $ U_{E} $ is defined as follows:
\begin{equation}
  U_{E}(x_{i}) \coloneqq U_{enta} \left( \bigotimes^{n}_{j = 1} R_{Y}(x_{i}^{(j)}) \right) U_{enta} \left( \bigotimes^{n}_{j=1} R_{Y}(x_{i}^{(j)}) \right),
\end{equation}
where $ U_{enta} $ is an entanglement part composed of CNOT gates.
$ O $ is the measurement operator defined by $ O \coloneqq Z(0) $, and $ V $ (target unitary) is defined by
\begin{equation}
  V \coloneqq \prod_{l=1}^{2} \left( R_{Z}(\beta^{l}) \, R_{Y}(\gamma^{l}) \, R_{Z}(\nu^{l}) \right)^{\otimes n},
\end{equation}
where $ \beta $, $ \gamma $, and $ \nu $ are sampled uniformly from the range $ [0, 2\pi) $.
For the verification, we prepared the training dataset with size 10 ($ M = 10 $) and the test dataset with size 100.
The number of each dataset was 20.

\noindent
{\bf Model (assumption) construction.}
In this study, we adopted the transverse field Ising model (TFIM), because TFIM has different DLA dimensions according to different boundary conditions (open and closed).
The Hamiltonian of TFIM is defined as follows \cite{larocca2022}:
\begin{equation}
  H = \sum_{i = 1}^{n_{f}} Z_{i} Z_{i+1} + \sum_{i = 1}^{n} X_{i},
\end{equation}
where $ n_{f} = n - 1 $ for the open boundary condition and $ n_{f} = n $ for the closed boundary condition.
The generators ($ \mathfrak{g} $) of TFIM is $ \left\{ \sum_{i = 1}^{n_{f}} Z_{i} Z_{i+1}, \sum_{i = 1}^{n} X_{i} \right\} $.
Each DLA dimension is $ \dim(\mathfrak{g}) = n^{2} $ for the open and $ \dim(\mathfrak{g}) = n $ for the closed boundary conditions \cite{larocca2022}.

Corresponding to the target unitary $ V $, the assumption model $ U $, a parameterized unitary, is defined as follows:
\begin{equation}\label{eq21}
  U(\theta) = \prod_{l = 1}^{L} \prod_{k = 1}^{K} \exp(i \theta_{l,k} \mathfrak{g}),
\end{equation}
where $ \mathfrak{g} $ is a linear combination with all coefficients equal to one, which is composed of the generators of TFIM.
We note that $ y = \braket{ 0^{n} |\, U_{E}(x)^{\dagger} U(\theta)^{\dagger} O U(\theta) U_{E}(x) \,| 0^{n} } $.
In the experiments, $ L = 2 $ and $ K = 10 $, thus the number of parameters was 20.

To investigate the dependency of training algorithm $ \mathcal{A} $, we employed two distinct approaches: the parameter shift rule \cite{schuld2019} and random search.
For the parameter shift rule, with the simultaneous perturbation stochastic approximation (SPS) \cite{gacon2021}, the initial parameter values were drawn from a uniform distribution within the range $ [-0.01, 0.01) $.
In contrast, for the random search algorithm (RAN), the initial parameter values were set to zero.
The total number of epochs was 200.

\section{Results}
In this section, we present the numerical simulation results to confirm our main findings as follows:
\begin{enumerate}
\item Training results of the model using SPS and RAN to examine the differences between the training algorithms.
\item Generalization gap with varying numbers of qubits to verify the validity of Theorem \ref{thm1}.
\item Upper bound on the number of trainable parameters to verify the validity of Corollary \ref{coro1}.
\end{enumerate}

\subsection{Training results}
At first, the training results (the root mean squared errors (RMSEs)) for the training algorithms under the open and closed boundary conditions of the TFIM are shown in Figures \ref{fig-eee} and \ref{fig-fff}.
\begin{figure}[htbp]
  \centering
  \begin{tabular}{c}
    \includegraphics[width=8cm]{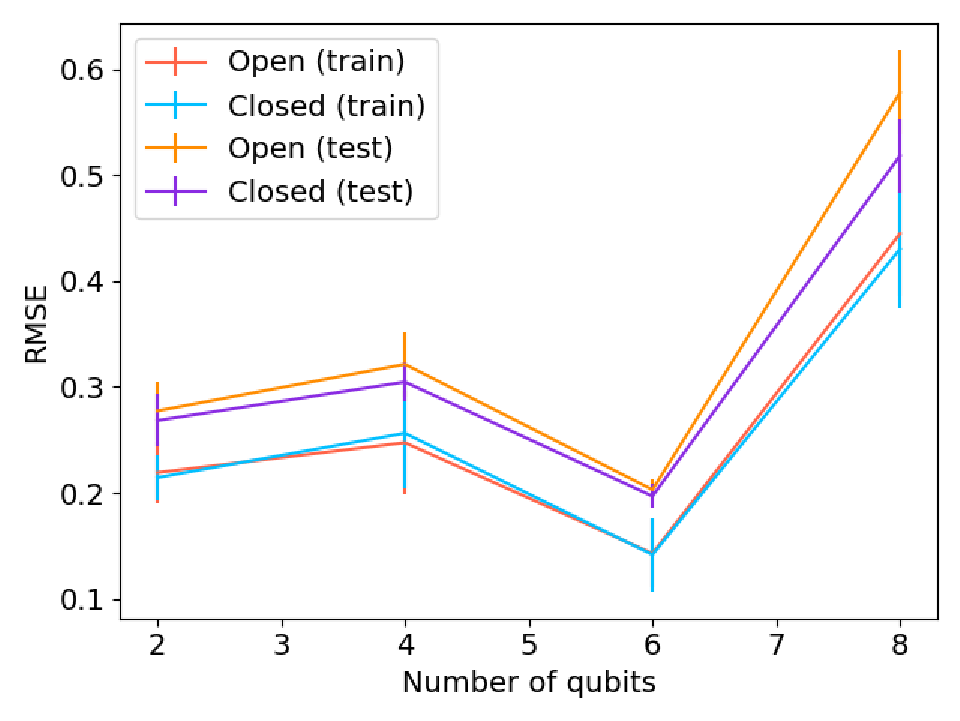}
  \end{tabular}
  \caption{RMSE results of SPS for training and test data}\label{fig-eee}
\end{figure}
\begin{figure}[htbp]
  \centering
  \begin{tabular}{c}
    \includegraphics[width=8cm]{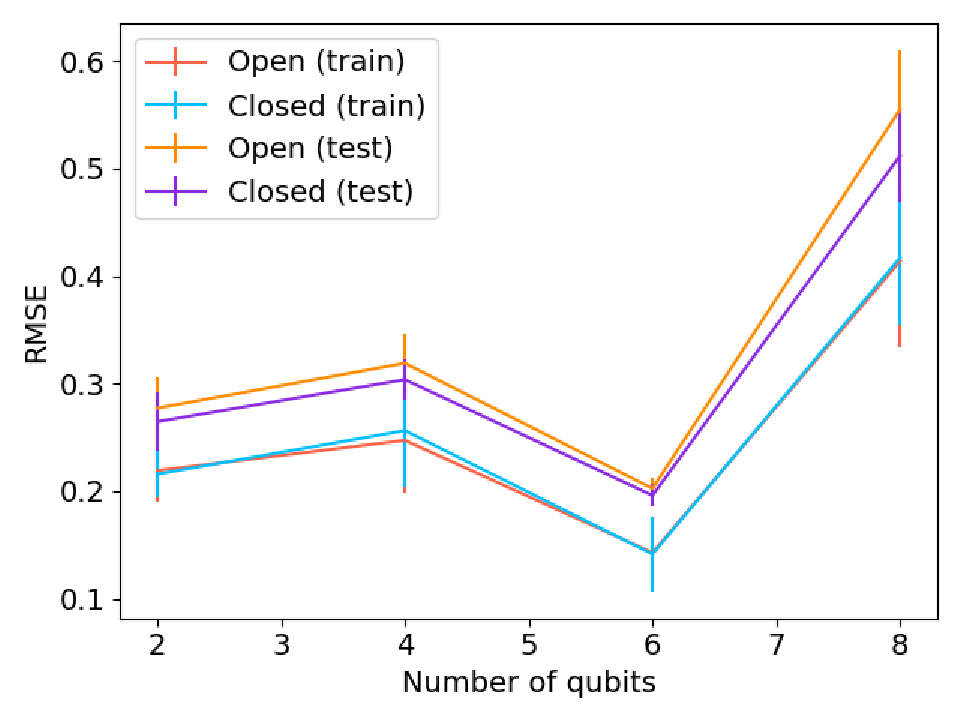}
  \end{tabular}
  \caption{RMSE results of RAN for training and test data}\label{fig-fff}
\end{figure}
In the figures, for example, the term ``Open (train)'' denotes the results of the open boundary condition for the training data.
Error bars indicate one standard deviation.
As the number of qubits increases, the RMSEs for both the training and test data increase.
No significant differences were almost observed due to the boundary conditions or the training algorithms\footnote{
  Among the 32 combinations for the two-sided Student's t-test, the P-value for the case of Open (test) vs. Closed (test) for SPS was 0.003 ($ < 1\% $).
  For all other cases, the P-values were greater than $ 5\% $.
}.
Therefore, we conclude that after training, the circuits obtained from SPS and RAN were almost equivalent with respect to the training.
However, for both training algorithms, the RMSEs for eight qubits were somewhat large.
This may be because the approximation ability of the circuit $ U $ (Eq. \ref{eq21}) is low.
Regarding the model type of $ U $, the realized non-linearity in the input-output relationship depends on both the input-data encoding circuit and the measurement.
Hence, such non-linearity may be sufficiently large for the training task involving eight qubits.

\subsection{Generalization gap}
Figures \ref{fig-yyy} and \ref{fig-zzz} show the results of averaged generalization gaps for the open and closed boundary conditions using SPS and RAN exploited as $ \mathcal{A} $, respectively.
\begin{figure}[htbp]
  \centering
  \begin{tabular}{c}
    \includegraphics[width=8cm]{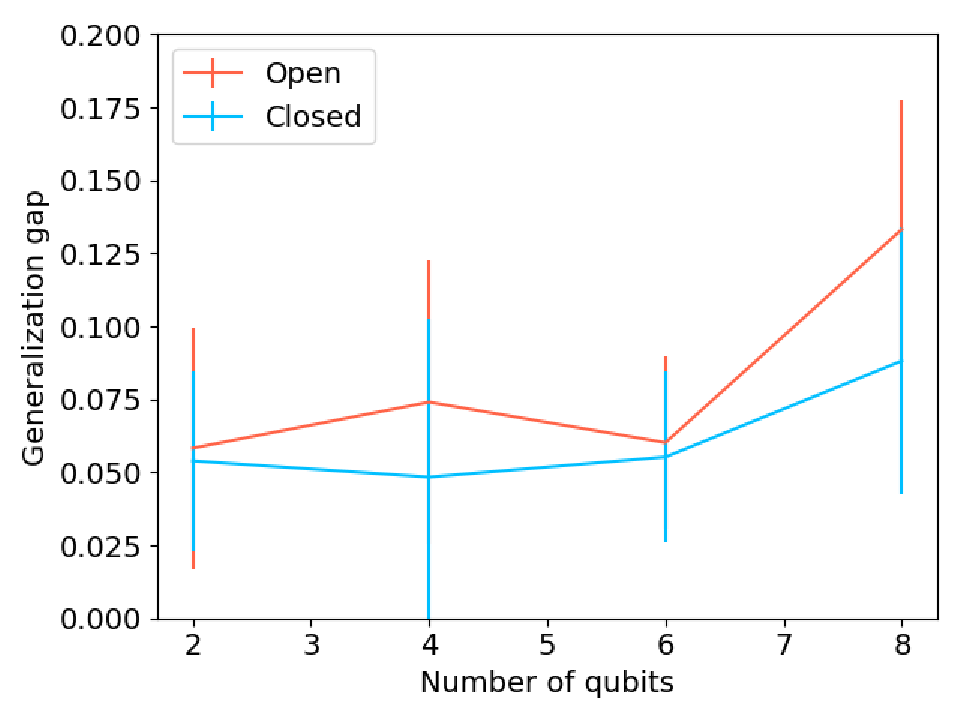}
  \end{tabular}
  \caption{Experimental results of generalization gaps. Training algorithm is SPS.}\label{fig-yyy}
\end{figure}
\begin{figure}[htbp]
  \centering
  \begin{tabular}{c}
    \includegraphics[width=8cm]{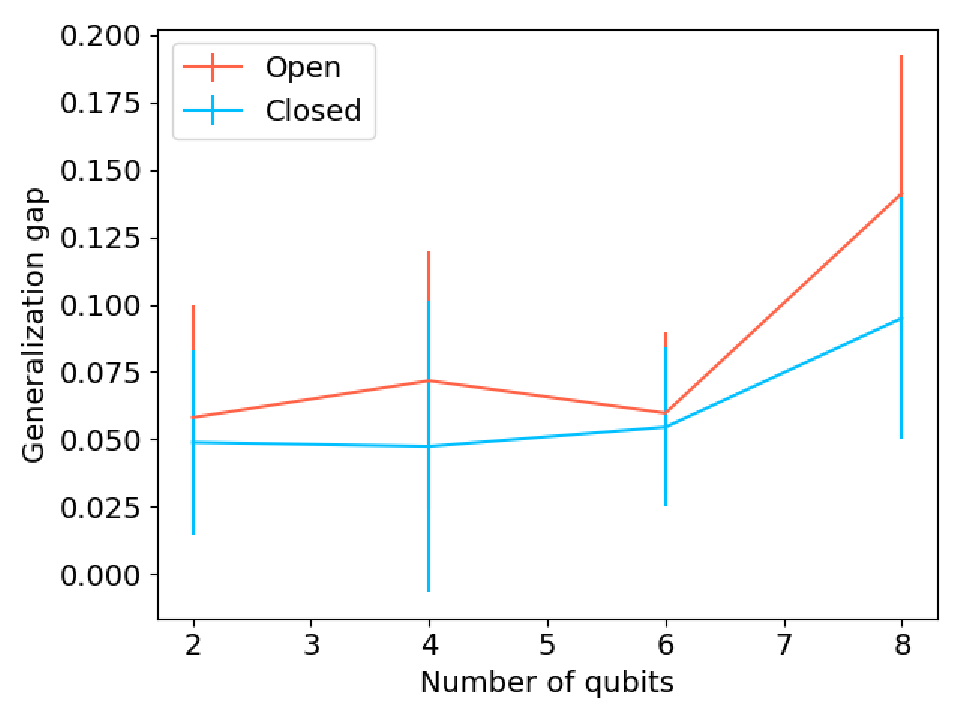}
  \end{tabular}
  \caption{Experimental results of generalization gaps. Training algorithm is RAN.}\label{fig-zzz}
\end{figure}
In the training experiments, we observed some instances where empirical risks did not decrease sufficiently due to random seed dependency.
In these cases, there are cases where generalization gaps are negative.
Therefore, we included only positive generalization gaps in the results.

From Figures \ref{fig-yyy} and \ref{fig-zzz}, both training algorithms exhibited similar tendencies in the results.
We also observed that increasing the number of qubits leads to larger generalization gaps, regardless of the training algorithm.
The variances in the generalization gap were relatively large.
Additionally, the generalization gaps for the open boundary condition were greater than those for the closed boundary condition.
It appears that the difference between the generalization gaps of the open and closed boundary conditions becomes more pronounced as the number of qubits increases.
This is suggested by Theorem \ref{thm1}, because for the open boundary condition, the generalization gap is scaled by $ \mathcal{O}(n) $, while for the closed boundary condition, that is $ \mathcal{O}(\sqrt{n}) $.
Regarding the linear approximation, the slopes of each results are shown in Table \ref{tab-aaa}.
\begin{table}[htb]
  \centering
  \caption{Results of slopes of generalization gap}\label{tab-aaa}
  \vspace{5pt}
  \begin{tabular}{ccc}\hline
    Algorithm & Boundary condition & Slope\\ \hline
    SPS & Open & 0.011\\
    SPS & Closed & 0.005\\
    RAN & Open & 0.012\\
    RAN & Closed & 0.007\\ \hline
  \end{tabular}
\end{table}
However, for the fitting curves for the data points using the linear approximation formula, $ {\rm R}^{2} $ values were approximately 0.6.
Therefore, the fitting curves did not sufficiently align with those predicted by Theorem \ref{thm1}.
Further investigation into these discrepancies is needed.

\subsection{Upper bound of the number of the trainable parameters}
Next, we evaluated whether the inequality Eq. \ref{eq18} in Corollary \ref{coro1} holds for the number of trainable parameters ($ N_{t} $).
The evaluation index (CR) was defined as follows:
\begin{equation}
  {\rm CR} \coloneqq \frac{\sum_{k=1}^{N_{t}} {\rm I}(N_{t} < N_{t}^{*})}{N_{t}},
\end{equation}
where $ N_{t}^{*} = \frac{2}{(2 - \exp(p))p} + 1 $ and $ p = \sigma_{max}(\theta_{l,k}^{*} \mathfrak{g}) $ ($ \theta_{k}^{*} $ denotes the $ l $-the layer and $k$-th optimized parameter values).
It should be noted that in this study, the operator norm of $ \theta_{k}^{*} \mathfrak{g} $ was used to estimate the upper bound of the operator norm of the generators after training.
However, instead of this generator setting, if an Hermitian matrix of $ \prod_{l = 1}^{L} \prod_{k = 1}^{K} \exp(i \theta_{l,k} \mathfrak{g}) $ is derived, the Hermitian matrix may be used for a more accurate estimation of the upper bound.

As the second index, the maximum values of $ p $ is defined as follows:
\begin{equation}
  p_{max} \coloneqq \max_{1 \leq l \leq L, 1 \leq k \leq K} \sigma_{max}(\theta_{l,k}^{*} \mathfrak{g})
\end{equation}
And the corresponding estimated $ N_{max} $ is defined as $ N_{max} \coloneqq \frac{2}{(2 - \exp(p_{max}))p_{max}} + 1 $. 

Figure \ref{fig-aaa} shows the results of CR for the training algorithms (SPS and RAN) with the open and closed boundary conditions.
\begin{figure}[htbp]
  \centering
  \begin{tabular}{c}
    \includegraphics[width=8cm]{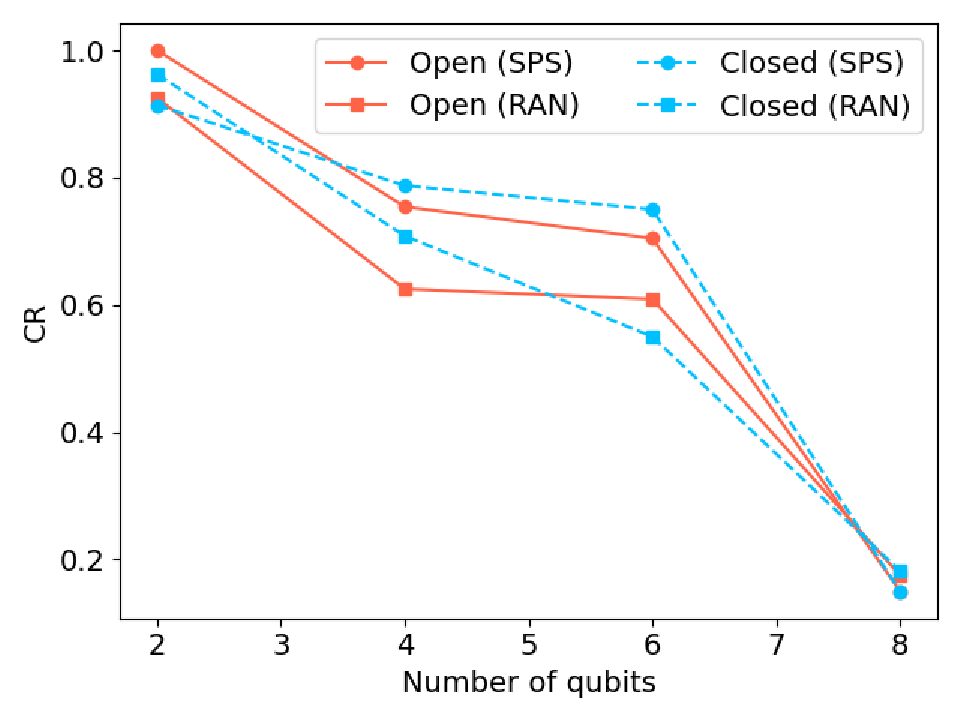}
  \end{tabular}
  \caption{CR results for open and closed boundary conditions. Training algorithms are a parameter shift (SPS) and a random search (RAN) algorithms. the solid lines denote the open boundary condition. The dashed lines denote the closed boundary condition.}\label{fig-aaa}
\end{figure}
In the figure, the data points represent averaged values.
Note that $ N_{t} = 20 $.
From the figure, we observed that increasing the number of qubits results in lower CR values.
For this observation, it remains unclear whether the underlying reasons are empirical or theoretical.

Figure \ref{fig-ggg} shows the results of $ p_{max} $ and $ N_{max} $ for SPS and RAN under both open and closed boundary conditions.
As the number of qubits increases, $ p_{max} $ increases while $ N_{max} $ decreases, regardless of the training algorithm or boundary condition.
We did not find clear differences between the results for open and closed boundary conditions.
However, there were differences between the results for SPS and RAN.
For the $ N_{max} $ result of SPS for eight qubits, the variance was large.
This is because the $ p_{max} $ values were in the vicinity of $ \ln(2) $ (Figure \ref{fig-ggg} (a)).
\begin{figure}[htbp]
  \centering
  \begin{minipage}{14.5cm}
    \SetFigLayout{2}{2}
    \subfigure[$p_{max}$ (SPS)]{\includegraphics[width=7cm]{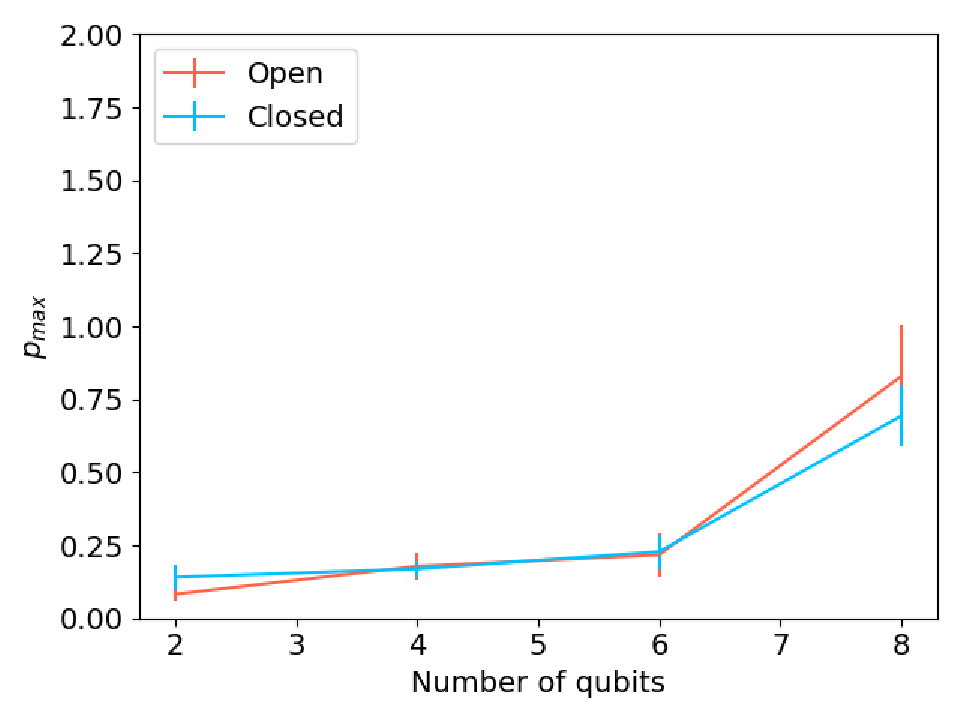}}
    \hfill
    \subfigure[$N_{max}$ (SPS)]{\includegraphics[width=7cm]{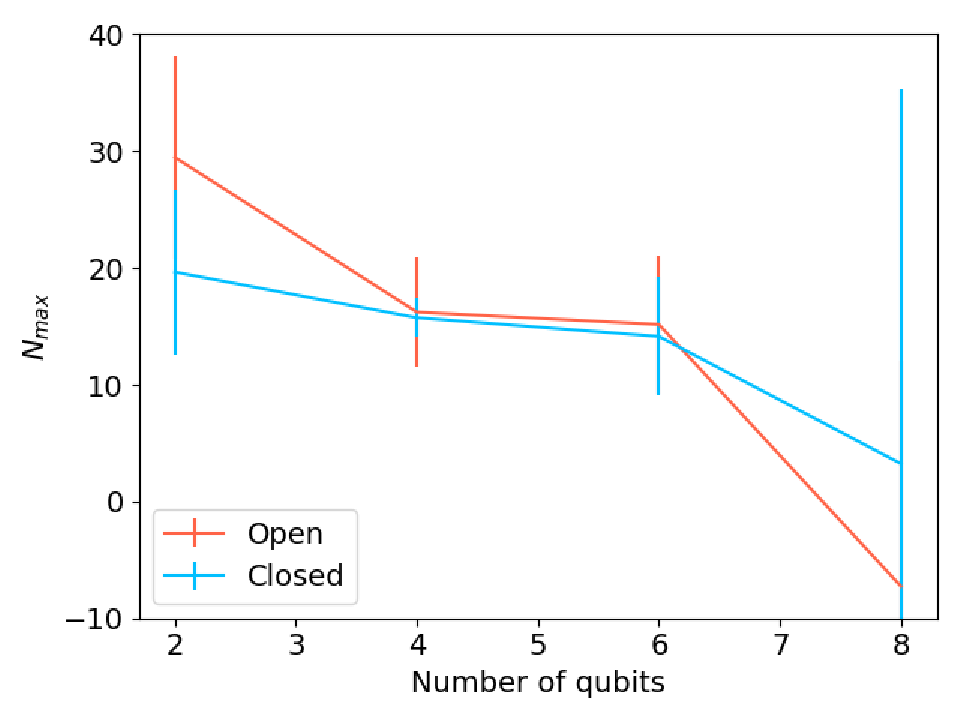}}\\
    \subfigure[$p_{max}$ (RAN)]{\includegraphics[width=7cm]{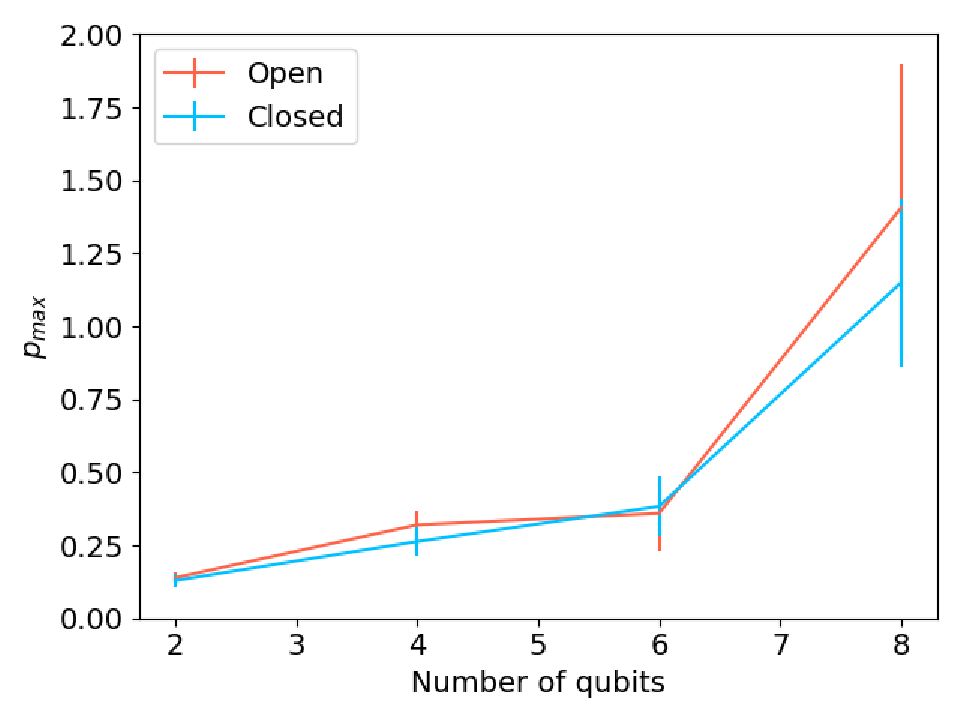}}
    \hfill
    \subfigure[$N_{max}$ (RAN)]{\includegraphics[width=7cm]{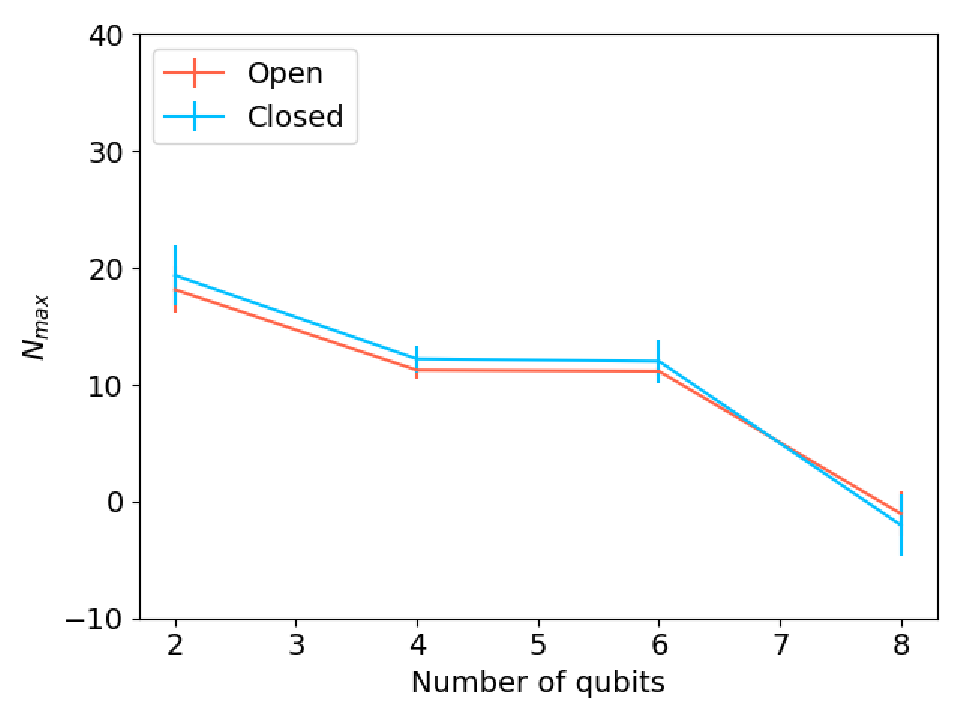}}
  \end{minipage}
  \caption{Results of $ p_{max} $ and $ N_{max} $ for SPS and RAN}\label{fig-ggg}
\end{figure}

In summary of the results related to Corollary \ref{coro1}, the index values suggested that satisfying the condition $ \exp(p) < 2 $ becomes difficult as the number of qubits increases.

\section{Discussion}
In this section, we explore desirable parameterized unitaries (models) in quantum machine learning from the perspective of DLA, addressing BPs, overparameterization, and generalization capability.

According to \cite{ragone2024}, if $ \dim (\mathfrak{g}) \in \mathcal{O}({\rm poly}(n)) $, the circuits do not exhibit BPs when the $ \mathfrak{g} $-purities (a projection map of a quantum state onto $ {\rm span}_{\mathbb{C}} ( \mathfrak{g} ) $, where $ \mathfrak{g} $ is a Lie algebra) for initial quantum state ($ \rho $) and measurement operator ($ O $) are $ \notin \mathcal{O}(1/b^{n}) $ with $ b > 2 $.
Larocca et al \cite{larocca2023} showed that when $ \dim (\mathfrak{g}) \in \mathcal{O}({\rm poly}(n)) $, quantum neural networks such as $ U $ (Eq. \ref{eq5}) exhibit overparameterization.
In this study, for $ U $, the generalization capability with respect to increasing the number of qubits is desirable when $ \dim(\mathfrak{g}) \in \mathcal{O}( n ) $.
Therefore, the condition that satisfies all three properties, avoiding BPs and ensuring overparameterization and better generalization, is $ \dim(\mathfrak{g}) \in \mathcal{O}(n) $.
Hence, a sufficient condition for desirable unitaries is that the skew-Hermitian matrices generating these unitaries have a DLA dimension of $ \mathcal{O}(n) $.
Such a unitary $ U $ for the $ n $-qubit system, which may serve as a simpler example, can be defined as follows:
\begin{equation}
  U \coloneqq R_{y} (\theta_{1}) \otimes R_{y} (\theta_{2}) \cdots \otimes R_{y} (\theta_{n}).
\end{equation}
We note that $ R_{y} (\theta) = \exp( -i \frac{\theta}{2} \sigma_{y} ) $, where $ \sigma_{y} $ is the Pauli Y matrix (Hermitian matrix).
Thus, when $ \theta_{k} $ are independent parameters, the DLA dimension becomes $ n $.

\section{Conclusion}
This study presented a generalization bound (Theorem \ref{thm1}) for quantum neural networks using covering numbers derived from DLA.
According to Theorem \ref{thm1}, we recommend using DLA with $ \dim(\mathfrak{g}) = n $ for better generalization capability.
Additionally, an upper bound (Corollary \ref{coro1}) for the trainable parameters is derived using the maximum operator norm of the generators.
Numerical simulation results for confirming the proposed theorem weakly supported the validation of the upper-bound order of the generalization gap.
Furthermore, we observed that for our generator setting, the upper bound for the trainable parameters does not hold as the number of qubits increases.
Regarding model training, the RMSE results exhibited relatively large values, which may affect the generalization gap.
We also observed that the variances in the generalization gap were large.
Therefore, it is inadvisable to perform statistical tests between the open and closed boundary conditions.
Further investigation is needed to address these issues.
Additionally, satisfying the upper bound for the trainable parameters may be difficult in practical.

Future work is needed to gain a further understanding of these results.
Additionally, we will explore whether the proposed theorem extends to other quantum neural networks, such as those with data re-uploading circuits \cite{perezsalinas2020}.

\bibliographystyle{plain}
\bibliography{references}

\end{document}